\pdfoutput=1
\newif\ifFull
\Fulltrue
\ifFull
\documentclass[11pt]{llncs}
\usepackage[margin=1in]{geometry}

\else
\documentclass{llncs}  
\renewcommand{\subsection}[1]{\paragraph{\bf #1.}}
\fi

\usepackage{graphicx}
\usepackage{cite}
\usepackage{amsmath,amsfonts,amssymb,latexsym}
\usepackage{parskip}

\begin{document}
\ifFull
\pagestyle{plain}
\fi

\title{Knuthian Drawings of Series-Parallel Flowcharts}

\author{Michael T. Goodrich \and Timothy Johnson \and Manuel Torres}

\institute{
Dept. of Computer Science, University of California, Irvine, CA USA 
}

\date{}

\maketitle

\begin{abstract}
Inspired by a classic paper by Knuth, 
we revisit the problem of drawing flowcharts of loop-free algorithms,
that is, degree-three series-parallel digraphs.
Our drawing algorithms show that it is possible to produce 
Knuthian drawings of degree-three series-parallel digraphs with
good aspect ratios and small numbers of edge bends.
\end{abstract}

\section{Introduction}
In 1963, Knuth published the first paper on a 
computer algorithm for a graph drawing problem,
entitled ``Computer-drawn Flowcharts''~\cite{Knuth:1963}.
In this paper, Knuth describes an algorithm that takes as 
input an $n$-vertex directed graph $G$ that represents a flowchart and, using the modern language
of graph drawing, produces an \emph{orthogonal drawing} of $G$,
where vertices are assigned to integer grid points and edges are polygonal
paths of horizontal and vertical segments.
In Knuth's algorithm, every vertex is given the same $x$-coordinate and
every edge has at most $O(1)$ bends, so that
drawings produced using his algorithm can be output line-by-line on an 
(old-style) ASCII line printer 
and have worst-case area at most $O(n^2)$.
Some drawbacks of his 
approach are that his drawings can be highly non-planar, even
if the graph $G$ is planar, and his drawings 
can have very poor aspect ratios, since every vertex is drawn along a vertical
line.
Nevertheless, his drawings possess an additional desirable
property that has not been specifically addressed since the time of his
seminal paper, which we revisit in this paper.

Specifically, inspired by his drawing convention, we say that 
a directed orthogonal graph drawing is \emph{Knuthian} if 
there is no vertex having an incident edge 
locally pointing upwards unless that vertex is a \emph{junction} node, 
that is, a vertex having in-degree strictly greater than its out-degree.
In other words, a directed orthogonal graph drawing is Knuthian
if no non-junction node has an in-coming edge into its
bottom or an out-going edge out of its top.
This property (rotated 180 degrees) is related to previously-studied concepts 
known as ``upward'' or ``quasi-upward'' drawing 
conventions~\cite{Chan2002153,quasi-99,ggt-put-96},
where all edges 
must locally enter a vertex from below and leave going up.
Intuitively, Knuthian drawings have a natural top-to-bottom flow 
through non-junction nodes while allowing for natural reverse-directional
flow through junction nodes.

A Knuthian drawing is different from upward and quasi-upward
drawing conventions in that a Knuthian drawing places no
orientation constraints on the edges incident to junction nodes.
In other words, a Knuthian
drawing is a hybrid between a quasi-upward orthogonal drawing 
(rotated 180$^\circ$) and an unconstrained orthogonal drawing,
and this hybrid nature raises some interesting graph drawing questions.

\paragraph{Our Results.}
In this paper, we describe efficient algorithms for producing Knuthian
drawings of degree-three series-parallel directed graphs,
which are equivalent to the flowcharts of loop-free algorithms.
We provide a recursive linear-time 
algorithm for producing such drawings of degree-three 
series-parallel graphs and 
we show that such a graph with $n$ vertices has a Knuthian drawing
with width $O(n)$ and height $O(\log n)$. 
We then show 
how to ``wrap'' this drawing, while still maintaining it to be
Knuthian, to fit within a fixed width,
so that the area is $O(n \log n)$ and the aspect ratio is constant.
Our drawings strive to achieve few edge bends, both in the
aggregate and per edge.
Our drawing approach contrasts with previous approaches to drawing
series-parallel graphs, including the standard recursive split-join-and-compose
method and Knuth's original method~\cite{Knuth:1963}, as well as more
recent methods for drawing series-parallel graphs
(e.g., see~\cite{b-small-11,doi:10.1142/S0218195994000215,Hong2000165}).
For example,
we are able to match the two-bends-per-edge bound achieved
by Biedl~\cite{b-small-11} while at the same time improving the area of the 
drawing, by exploiting the area improvements possible using the Knuthian
drawing convention.
Admittedly, our Knuthian drawings are not always upward, but 
they nevertheless achieve an ``upward-like'' quality 
through the careful use of junction nodes.

\section{Classic Algorithms}
Two classic algorithms
for drawing series-parallel flowcharts are 
the standard split-join-and-compose algorithm
and Knuth's algorithm~\cite{Knuth:1963}.
We review both of these algorithms in this section.

In the standard split-join-and-compose
algorithm for drawing a series-parallel flowchart,
we draw the graph recursively according to recursive structure of
the graph in a left-to-right or top-down fashion (w.l.o.g., let us assume
a left-to-right orientation).
Two graph components formed as a parallel construction are drawn recursively
along two side-by-side horizontal strips, which are connected to a split node with in-coming edge
from the left and outgoing vertical edges that 
each bend at the in-coming segment
for the strip. Each outgoing edge is drawn as a horizontal line that bends to 
join back at a junction node having the same $y$-coordinate as the split node.
Series constructions are simply composed left-to-right and connected by
horizontal segments.
Such graphs are degree-three series-parallel graphs, because they
correspond to the control flow of a loop-free algorithm, which contains
sequences of statements and (possibly-nested) if-then-else decision 
constructs.
(See Figure~\ref{fig-s1} for two examples.)

\begin{figure}[htbp]
\begin{center}
\includegraphics[scale=0.4]{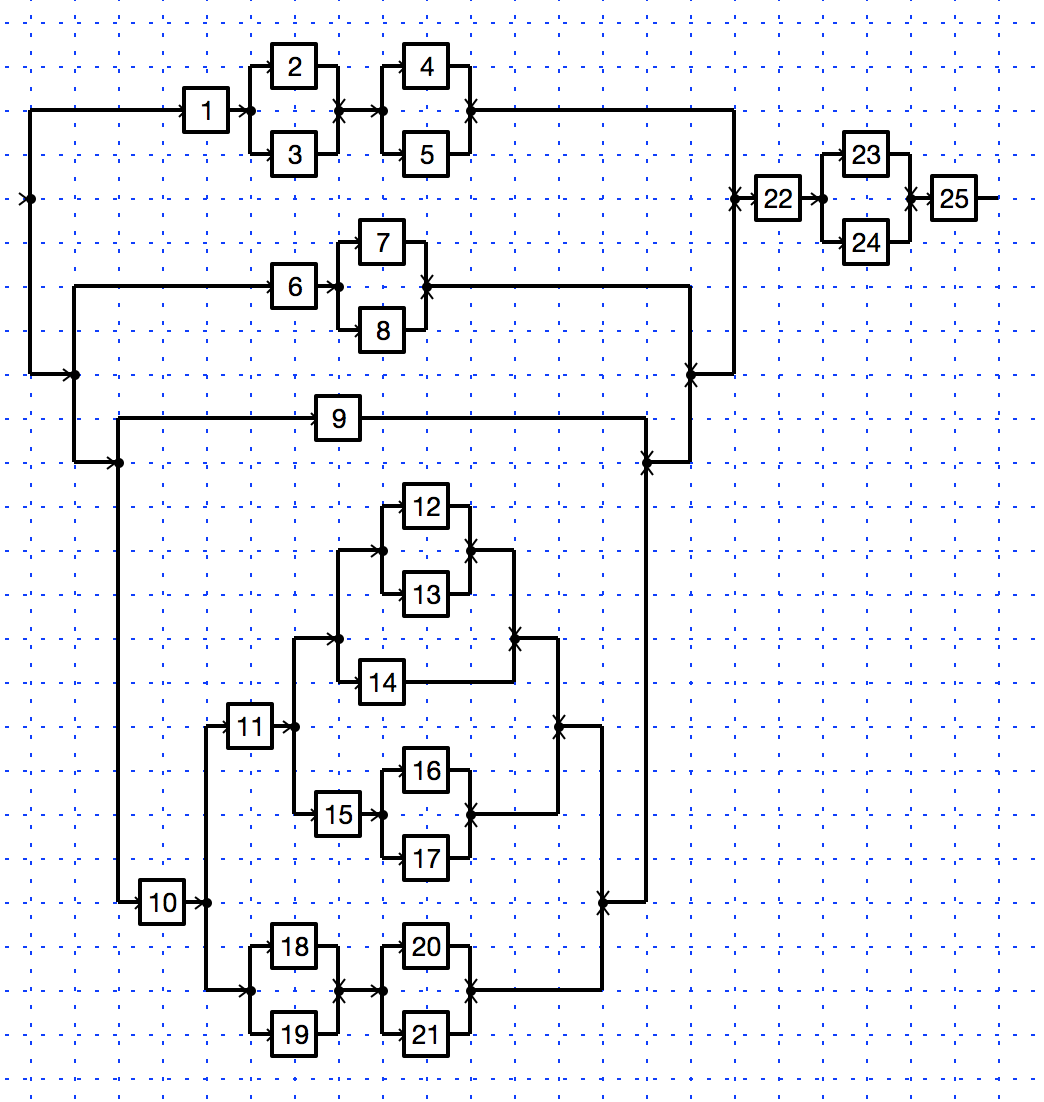} \\
(a) \\[12pt]
\includegraphics[scale=0.4]{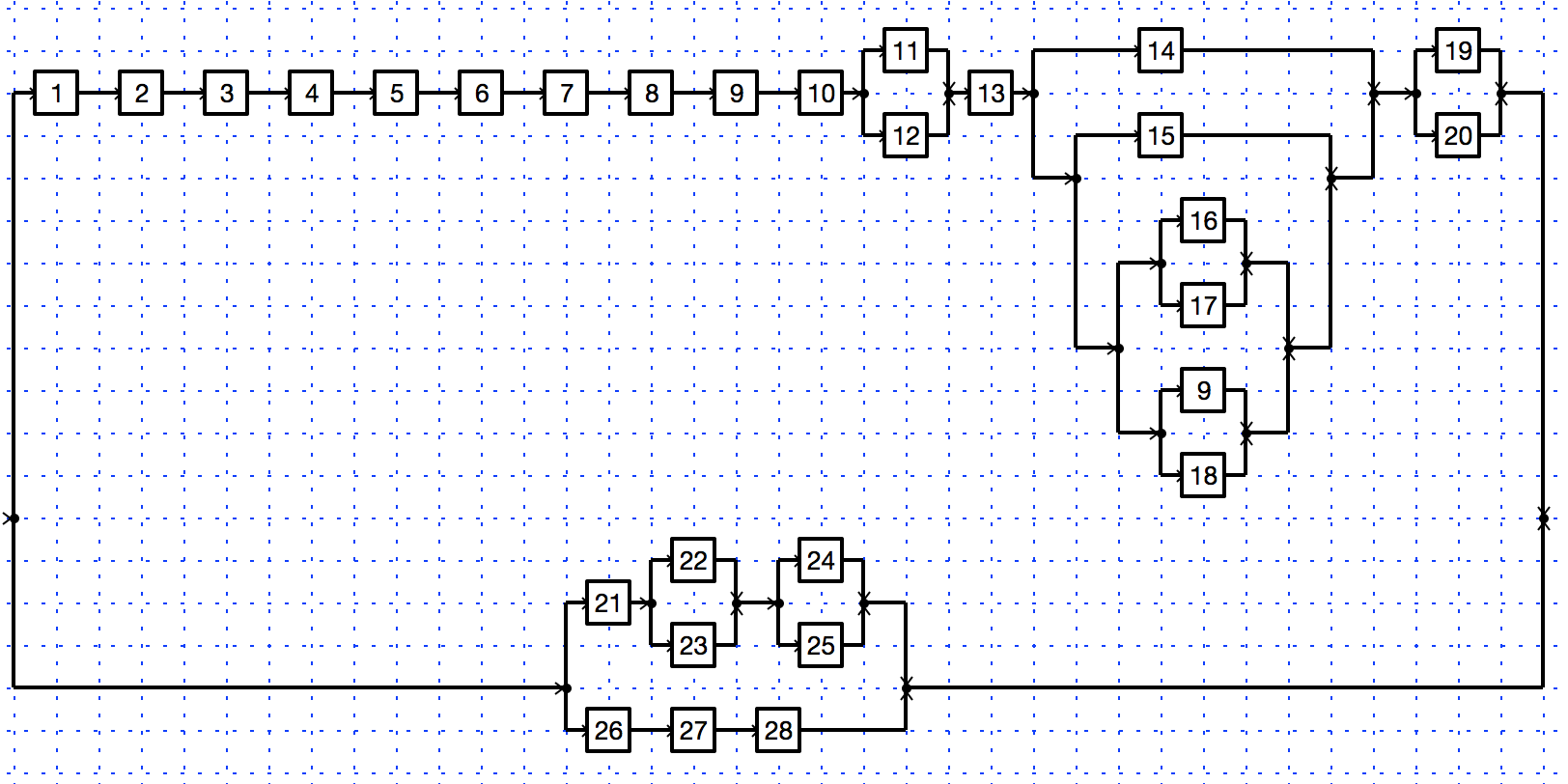} \\
(b)
\end{center}
\caption{Two graphs using the standard split-join-and-compose algorithm.
(a) Example 1; (b) Example 2.}
\label{fig-s1}
\end{figure}

In Knuth's algorithm, we place all of the nodes in a single column along the 
left side of the drawing (to all have the same $x$-coordinate) using some
reasonable ordering of the vertices. 
We then draw each edge in a top-to-bottom greedy fashion
using the space to the right of this column.
If an edge is between two consecutive nodes in the column, we simply connect
these two nodes with a single vertical segment.
Otherwise, we draw an edge $(v,w)$, with $w$ coming after $v$,
to go horizontally out from $v$ (optionally
after a first short downward segment from $v$), bending to go vertically
down to a row for $w$, and then extending horizontally back to $w$
(optionally adding a short vertical segment going into the top of $w$).
An edge $(w,v)$, with $v$ coming before $w$, is drawn in a 
similar fashion, except that the edge is drawn in reverse direction and the
optional short vertical segments would be coming into $v$
from above and going out from below $w$.
(See Figure~\ref{fig:test1} for two examples.)

\begin{figure}[htbp]
\begin{center}
\begin{tabular}{c@{\hspace{1in}}c}
\includegraphics[scale = 0.6]{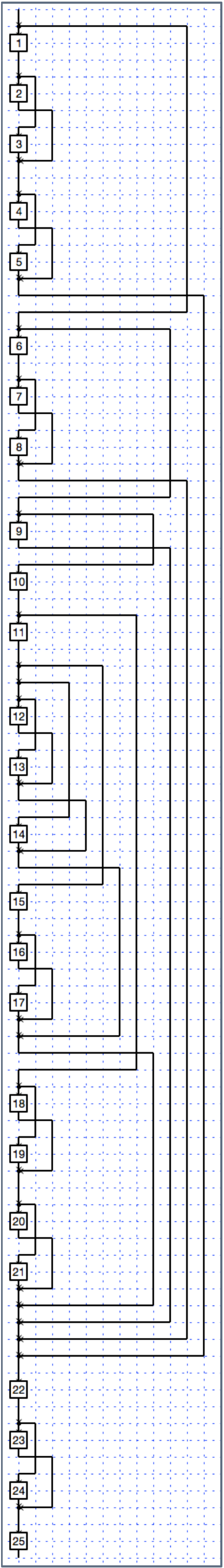}
&
\includegraphics[scale = 0.6]{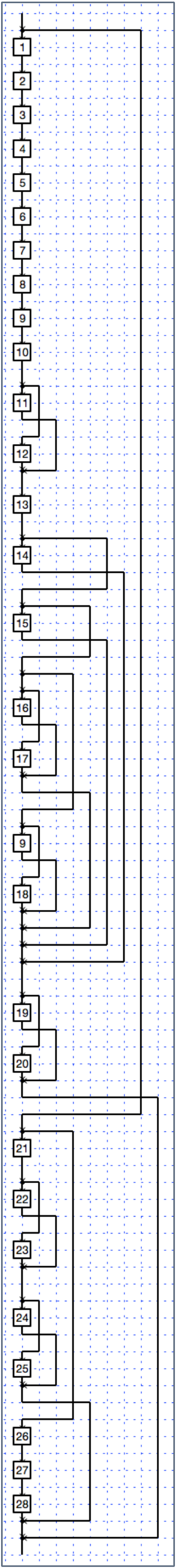} \\
(a) & (b)
\end{tabular}
\end{center}
	\caption{Two graphs drawn using Knuth's algorithm.
  (a) Example 1; (b) Example 2.}
  \label{fig:test1}
  \label{fig:test2}
\end{figure}

These algorithms both have the drawback of requiring $O(n^2)$ area
in the worst case.
For Knuth's algorithm, we will always have height $O(n)$. 
In the worst case, as exemplified in Figure~\ref{fig-w5}, 
we may have to use $O(n)$ distinct columns for edges, 
which may force the width to also be $O(n)$ in the worst case.
Alternatively, for large graphs using a small number of columns, the aspect
ratios of their drawings can be quite bad, as shown in Figure~\ref{fig:test1}.

\begin{figure}[htbp]
	\centering
	\includegraphics[scale=0.4]{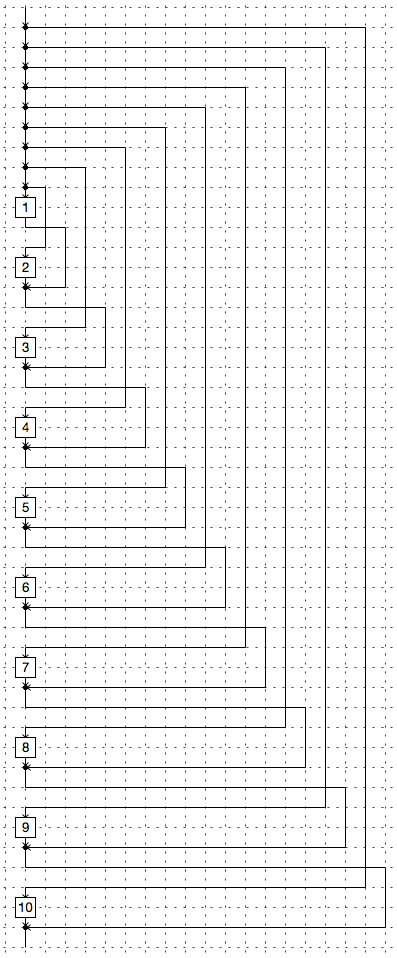}
	\caption{A worst-case quadratic-area example for Knuth's algorithm.}
        \label{fig-w5}
\end{figure}

To see that the standard split-join-and-compose algorithm
also uses $O(n^2)$ area, 
note that
we may have graph with $n/2$ nodes in series, followed by $n/2$ 
nodes in parallel. 
Then the height (in a left-to-right orientation) will be $O(n)$, 
because the $n/2$ parallel nodes will all be stacked on top of one another,
while the width will also be $O(n)$, because of the $n/2$ series nodes drawn
along the same horizontal strip.
(See Figure~\ref{fig-w6}.)

\begin{figure}[htbp]
	\centering
	\includegraphics[scale=0.6]{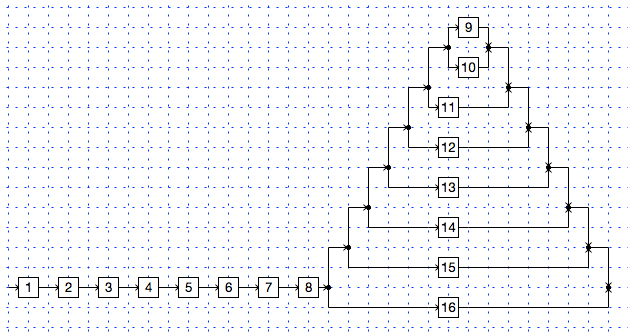}
	\caption{A worst-case quadratic-area example for the standard 
         split-join-and-compose algorithm.}
        \label{fig-w6}
\end{figure}

\section{Knuthian Drawings of Series-Parallel Flowcharts with $O(n\log n)$ Area}
In this section, we show that any $n$-vertex degree-three 
series-parallel graph has
a Knuthian drawing with $O(n \log n)$ area. 
In Appendix A, we discuss how to implement this drawing in linear time.
We achieve this area bound using a Knuthian drawing 
scheme that uses $O(n)$ width
and $O(\log n)$ height, and we show in 
the next section how to improve the aspect
ratios for such drawings while maintaining an $O(n\log n)$ area bound.

Let $G$ be a series-parallel graph with $n$ nodes.
$G$ begins with $s_1\ge 0$ nodes in series until it reaches
its first split node for two parallel subgraphs.
Likewise, $G$ ends with $s_2\ge 0$ nodes in series after its
last junction joining two parallel subgraphs.
We can categorize $G$ as being one of two types, which we respectively
call ``broad'' and ``pinched,''
as illustrated in Figure~\ref{fig3}.
In the broad case (the upper graph in Figure~\ref{fig3}), the first 
parallel split after the initial $s_1$ series nodes divides into 
series-parallel subgraphs of size $p_1$ and $p_2$, respectively,
and this parallel split is joined just before the last $s_2$ series nodes.
In the pinched case (the lower graph in Figure~\ref{fig3}), 
the first parallel split after the initial $s_1$ series nodes
(into subgraphs of size $p_1$ and $p_2$, respectively)
is not joined just before the last $s_2$ series nodes, and instead 
there is a subgraph of $p_3\ge 0$ nodes occurring
between a last parallel split (just prior to the last $s_2$ series nodes), 
which is divided into subgraphs of size $p_4$ and $p_5$. 

\begin{figure}[htb]
\centerline{\includegraphics[scale=1]{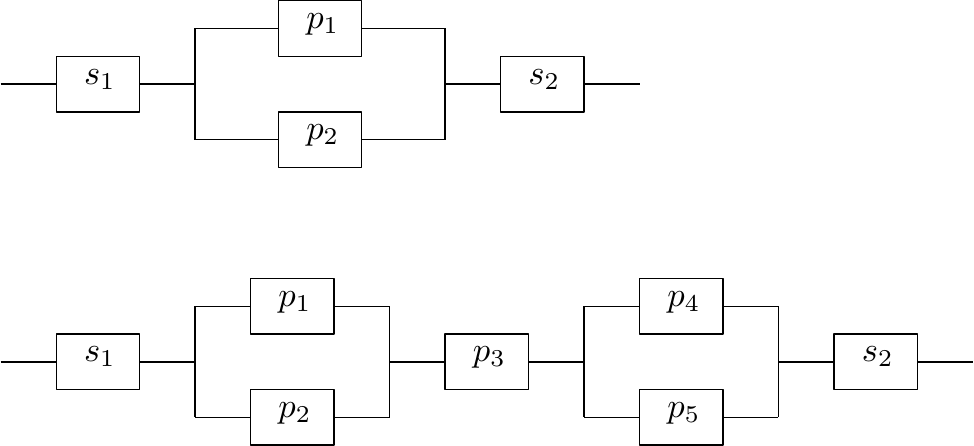}}
\caption{We categorize series-parallel graphs as one of 
these two types---``broad'' or ``pinched.''
Either the branches from the first decision node merge at the final junction node, or they merge sooner and the final junction node combines two branches from a later decision node.}
\label{fig3}
\end{figure}

Let us first consider the broad case.
We assume, without loss of generality, that $p_1 \geq p_2$. 
Then we draw the graph recursively as in Figure~\ref{fig4}, possibly with
one additional bend for the incoming and outgoing edges for the subgraph
of size $p_2$, for the incoming edge to enter in the upper left and 
the outgoing edge to exit in the lower left. 

\begin{figure}[htb]
\centerline{\includegraphics[scale=1]{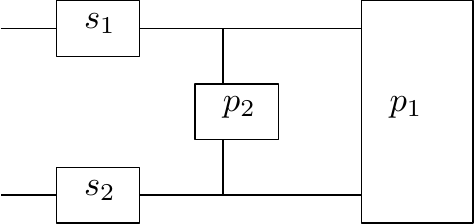}}
\caption{Recursive drawing for broad series-parallel graphs with width $O(n)$ and height $O(\log n)$.}
\label{fig4}
\end{figure}

Repeating this process results in a drawing pattern as illustrated
in Figure~\ref{fig1}.

\begin{figure}[htb]
\centerline{\includegraphics[scale=1]{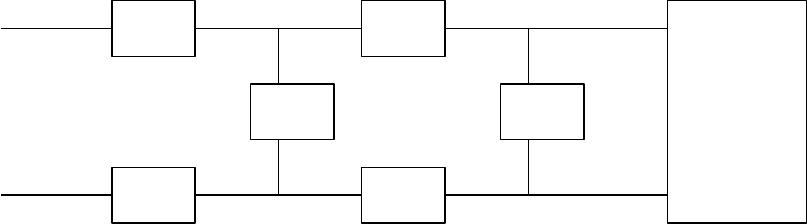}}
\caption{Our main pattern begins at the upper left and winds around to the lower left. 
}
\label{fig1}
\end{figure}

Let us next consider the pinched case.
We assume, without loss of generality, that $p_1 \geq p_2$ 
and that $p_4 \geq p_5$. 
Then we take the largest of $p_1$, $p_3$, and $p_4$, and 
recursively draw its corresponding
subgraph on the far right side, so that it can take up the full height of 
the drawing. 
Figure~\ref{fig5} shows how to recursively draw
each of the respective cases in which $p_1$, $p_3$,
or $p_4$ are the largest. 

\begin{figure}[htb]
\centerline{\includegraphics[scale=1]{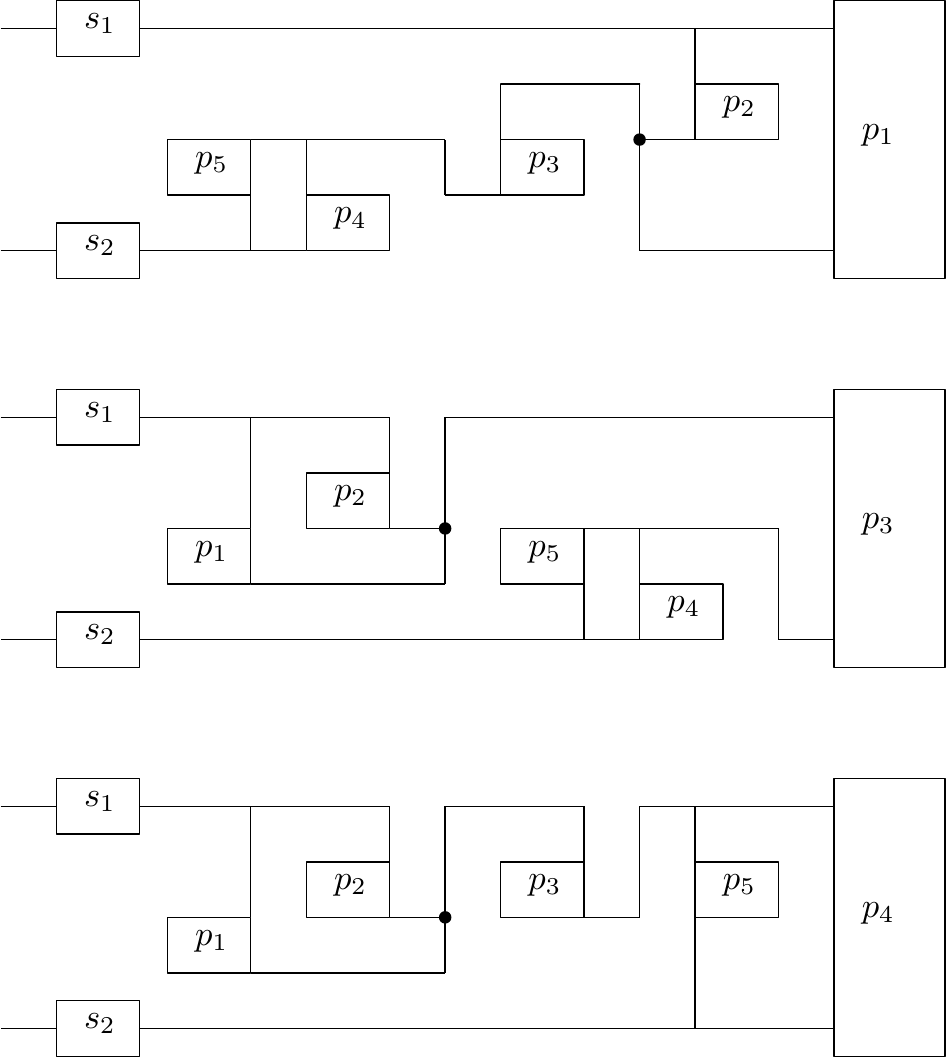}}
\caption{Recursive drawing patterns 
for pinched series-parallel graphs.
The first pattern is used when $p_1$ is the largest of $p_1$, $p_3$, and $p_4$,
the second pattern is used when $p_3$ is the largest, 
and the third pattern is used when $p_4$ is the largest.
The bold node in each pattern 
is an example of a junction node that takes advantage
of the freedom provided by a Knuthian drawing. Note that if $p_3=0$,
then we can simplify the middle component of each pattern to
be a straight line. If $p_3$ contains a chain of nodes in series with no parallel edges, then we can
draw all of the nodes along a straight line. This saves 2 bends in our pattern in either case.}
\label{fig5}
\end{figure}

% The smaller subgraphs that are on horizontal lines are drawn recursively using the left-to-right pattern in Figure 2, while the subgraphs on vertical lines and the larger subgraph on the right are drawn recursively in this same left-to-left pattern.

% \begin{figure}[h]
% \centerline{\includegraphics[scale=1]{figures/fig2.pdf}}
% \caption{We can modify the winding style given in Figure 1 so that it goes from left to right or from right to left, each of which only requires using a single extra row. We need separate drawings rather than taking a mirror image because the left-to-right nodes need to lie below the top edge, and the right-to-left nodes need to lie above the bottom edge.}
% \end{figure}

This algorithm results in the following.

\begin{theorem}
\label{thm-1}
A degree-three
series-parallel graph with $n$ vertices has a Knuthian drawing with width
$O(n)$ and height $O(\log n)$, such that each edge has at most two bends
and the total number of bends is at most $1.25n$.
\end{theorem}

\begin{proof}
Assume we have a degree-three series-parallel graph with $n$ nodes. 
Recall that
Figure~\ref{fig3} shows the two possible types of series-parallel graphs,
broad and pinched, that we consider. 
Let's take the first of these types, the broad case, where,
we assume, without loss of generality, that $p_1 \geq p_2$, and
we draw the graph recursively as in Figure~\ref{fig4}. 
In this case, the width, $W(n)$, is clearly $O(s_1+s_2+p_1+p2)$,
which is $O(n)$.
The required height, $H(n)$, in this case is 
$\max\{H(p_1),\, H(p_2) + 2\}$. Since $p_1 \geq p_2$, we have $p_2 \leq n/2$, 
so $H(n) \leq H(n/2) + O(1)$, when in this case. 

Now take the second case, of a pinched subgraph. 
Assume without loss of generality that $p_1 \geq p_2$ 
and that $p_4 \geq p_5$. 
Then we take the largest of $p_1, p_3, p_4$, 
and draw it on the far right side, 
so that it can take up the full height of the drawing, 
as shown in Figure~\ref{fig5}.
In each of these cases, it is easy to see that the width, $W(n)$,
is $O(n)$.
In the first of these cases, the total height, $H(n)$,
is at most
\[
\max\{H(p_1), H(p_2) + 2, H(p_3) + 2, H(p_4)+3, H(p_5) + 2\}.
\]
We know, in this case, that $p_2$ and $p_3$ are each at most $n/2$, 
and that $p_4 + p_5 \le 2n/3$. 
In addition, similar bounds hold for the second and third cases,
implying that
$H(n) \leq H(2n/3) + O(1)$, when in any of these cases. 
Thus, this bound applies independent of whether we are in the broad or
pinched cases.
Therefore, $H(n)$ is $O(\log n)$.

To establish the bound on the number
of bends per edge, notice that the number of bends
per edge is zero in our pattern for Figure~\ref{fig4}. 
There is also a base case (not shown), 
when we have only series nodes, which 
may introduce 1 bend (in the case of a single node).
The number of bends per edge
in any of our patterns in Figure~\ref{fig5} is at most two, even
considering a base case when $p_1=1$, $p_3=1$ or $p_4=1$. Moreover,
the in-coming edges into an upper 
left corner of a pattern can come from above or 
from the left, but we can ``slide'' an origin node in each case so that 
we don't introduce any new bends per edge no matter if we are entering
a pattern from the left or from above.
Likewise, a similar argument applies to the edge leaving a pattern that goes
out to the left or down.
Thus, the number of bends per edge in our final drawing is at most two.

To establish the bound on the total number of bends, notice that our 
pattern in Figure~\ref{fig3} does not introduce any bends, 
and the unshown base case of a single node introduces at most 1 bend.
Each
of our patterns in Figure~\ref{fig5} introduces at most 5 bends.
Moreover, each pattern in Figure~\ref{fig5} places 
at least 4 graph vertices into the drawing (even if $s_1=0$ and $s_2=0$).
Therefore,
we may account for all the bends we introduce by 
using an amortization scheme where we charge
each node in any of our patterns (including the unshown base case) 
for 1.25 bends.
This results in a total number of $1.25n$ bends.
\qed
\end{proof}

In an appendix, we show how this drawing can be computed in linear time.

\section{Fixed-Width Drawings}
In this section, we show how to adapt our 
$O(n \log n)$-area drawings, which admittedly have poor aspect ratios,
so that they achieve constant aspect ratios, proving the following theorem.

\begin{theorem}
A degree-three series-parallel graph with $n$ nodes has a Knuthian drawing
that can be produced in 
linear time to have width $O(A + \log n)$ and height $O((n/A)\log n)$,
for any given $A\ge\log n$; hence, the area is $O(n \log n)$.
The total number of bends is at most $3.5n + o(n)$.
\end{theorem}

\begin{figure}[htbp]
\begin{center}
\begin{tabular}{c|c|c|c}
& \multicolumn{2}{|c|}{\textbf{Rotated orientation adjustments}} \\ \hline
Node type & Previous orientation & $180^{\circ}$ rotation & Total bends added \\ \hline
Process & \includegraphics[scale = 1]{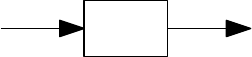} & \includegraphics[scale = 1]{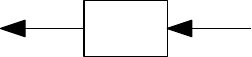} & 0 \rule{0pt}{4.0ex} \\ \hline
Process & \includegraphics[scale = 1]{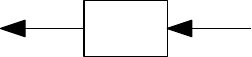} & \includegraphics[scale = 1]{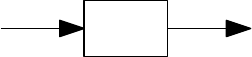} & 0 \rule{0pt}{4.0ex} \\ \hline
Process & \includegraphics[scale = 1]{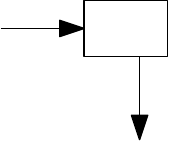} & \includegraphics[scale = 1]{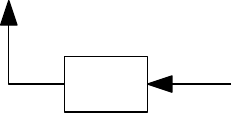} & 1 \rule{0pt}{10.0ex} \\ \hline
Process & \includegraphics[scale = 1]{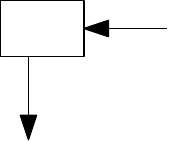} & \includegraphics[scale = 1]{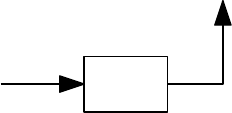} & 1 \rule{0pt}{10.0ex} \\ \hline
Process & \includegraphics[scale = 1]{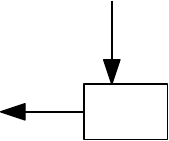} & \includegraphics[scale = 1]{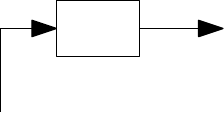} & 1 \\ \hline
Process & \includegraphics[scale = 1]{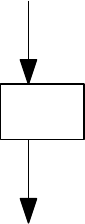} & \includegraphics[scale = 1]{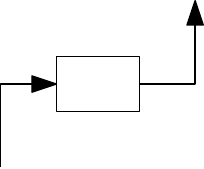} & 2 \\ \hline
Process & \includegraphics[scale = 1]{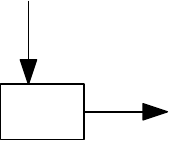} & \includegraphics[scale = 1]{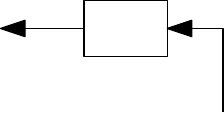} & 1 \\ \hline
Decision & \includegraphics[scale = 1]{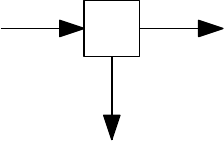} & \includegraphics[scale = 1]{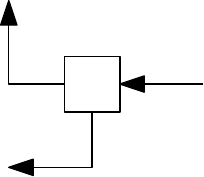} & 2 \\ \hline
Decision & \includegraphics[scale = 1]{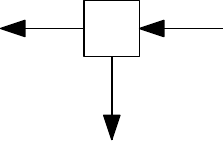} & \includegraphics[scale = 1]{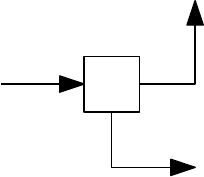} & 2 \\ \hline
Decision & \includegraphics[scale = 1]{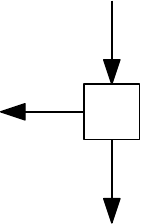} & \includegraphics[scale = 1]{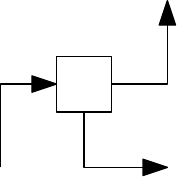} & 3 \\ \hline
Decision & \includegraphics[scale = 1]{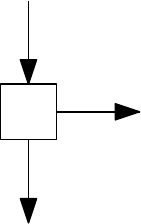} & \includegraphics[scale = 1]{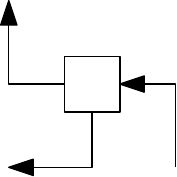} & 3 \\
\end{tabular}
\end{center}
\caption{The transformations necessary when rotating a Knuthian
drawing in order to maintain the Knuthian-drawing property. Note that we
do not need to transform junction nodes.}
\label{fig-kk}
\end{figure}

\begin{proof}
Our proof is based on 
taking our single-row drawing from Theorem~\ref{thm-1}
and performing a sequence of transformations that intuitively involve 
our dividing the drawing into ``slabs'' and then ``folding'' those slabs 
to achieve a good aspect ratio. 
The main details of the proof involve describing the changes necessary in
order to maintain the Knuthian-drawing property even as we rotate a slab
by $180^\circ$.

Let us assume, then, that we are given a drawing of our 
degree-three series-parallel graph as in Theorem~\ref{thm-1}.
First, we order all of the nodes by their $x$-coordinate, breaking ties by their $y$-coordinate. 
Then we take this list and divide it into slabs with exactly $A$ nodes (the last slab being somewhat smaller, if necessary).

Next, we stack these $n/A$ sections above one another, 
with every other section rotated $180^{\circ}$, and reconnect all of 
the edges that were cut,
as shown 
in Figure~\ref{fixedwidth}. 
Since our single-row drawing had height $O(\log n)$, there are at most $O(\log n)$ of these edges.

\begin{figure}[htb]
	\centering
	\includegraphics[scale = 1]{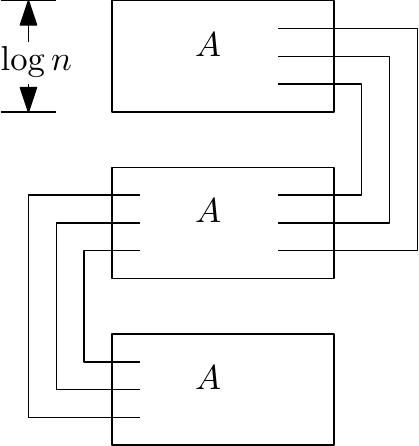}
	\caption{Wrapping our drawing into sections with $A$ nodes.}
	\label{fixedwidth}
\end{figure}

Lastly, we have to make sure that the nodes in each rotated section still satisfy our requirements for a Knuthian drawing. This requires adding bends to the incoming and outgoing edges at some of our decision and process nodes. The junction nodes do not have any restrictions on the locations of incoming or outgoing edges, so no changes to them are necessary.
The table shown in Figure~\ref{fig-kk} shows all of 
the possible connections, and how we choose to redraw them.
By inspection, it should be clear that performing these transformations
imply that the resulting drawing is Knuthian.
Adding the bends will increase the width of our graph, 
but only by a small constant factor.
The width of each slab's transformed section is $O(A)$, plus $O(\log n)$ for the extended edges. The height of each section is $\log n$, times $n/A$ sections.
This establishes the claimed width and height bounds for the theorem.

Finally, we prove the claim that our drawing has at most $3.5n + o(n)$ bends. Recall that our original drawing had at most $1.25n$ bends. Then we have added additional bends for the edges between sections, for the process nodes, and for the decision nodes.

First, we count bends for the edges between our sections of size $A$. There are $n/A$ of these sections, and $O(\log n)$ edges between them, each of which has two additional bends. This gives a total of $\frac{2n \log n}{A}$ additional bends. However, since we will choose $A$ to be larger that $\log n$, this will be a lower-order term, which we will ignore.
Second, we count bends for the process nodes. There are $n/2$ nodes in the rotated sections, and from our table we see that each could have up to two extra bends. This adds a total of up to $n$ bends.
Third, we count bends for the decision nodes. We charge the first decision node in each section to the previous recursive level, since the previous level determines from what direction we enter.

In the broad case, we ignore the first decision node and only count the two for the two subsections. The decision node for the larger subsection will take two bends, while the decision node for the smaller subsection could take three. This gives an average of $2.5$ bends per decision node.
In our diagram for the pinched case, three of the five sections have edges entering from the top, and will take three extra bends to fix when they are rotated. The other two will only take two bends. We then have two other decision nodes, but as stated, the first of these is already charged to the previous level. The other adds another two bends. This gives a total of 15 bends charged to six decision nodes, or once again, $2.5$ bends per decision node.
To finish, we need to count the total number of decision nodes in each section. This is bounded by the number of process nodes in the section, plus the difference between the number of edges leaving and entering the section. In our case, this is dominated by $A$, the number of process nodes in the section. Since half of our sections are rotated, we have at most $n/2 + o(n)$ rotated decision nodes, for a total of $1.25n$ extra bends.
So to sum up, we get $1.25n$ bends from the original graph, $o(n)$ new bends between the sections, $n$ new bends from the process nodes, and $1.25n + o(n)$ new bends from the decision nodes, for a total of $3.5n + o(n)$ bends.
\qed
\end{proof}

\section{Experimental results}

We implemented and tested our Knuth drawing algorithm algorithm
on some sample degree-three series-parallel graphs,
based on
two distributions used to create random binary series-parallel 
decomposition trees. 
The 
first case is defined by three parameters: $n$, 
the number of nodes; $p$, the probability that an internal node is a parallel composition (rather than series); and $s$, the expected fractional size of the left subtree.
To construct the tree, we first choose a root node, and set it to a parallel composition with probability $p$ and a series composition with probability $1 - p$. Then we sample $x$ from the normal distribution with mean $s \cdot n$  and standard deviation $n/20$. The size of the left subtree is then $\lfloor x \rfloor$, and the size of the right subtree is then $n - \lfloor x + 1 \rfloor$.
This process is then repeated recursively for the left and right subtrees. We reach our base case when $x < 1$, in which case we let the subtree consist of a single leaf node.
	
The second version is defined by just the first two parameters,
$n$ and $p$, as above. 
But instead of using a normal distribution to 
determine the size of each subtree, we use a uniform distribution.
	
We performed experiments by
running 300 random inputs through our algorithm for small-area drawing, 
Knuth's algorithm, and the standard split-join-and-compose algorithm 
for graphs of 10, 20, 50, 100, 200, 500, and 1000 nodes, respectively. 
We show the results of these experiments in Figures~\ref{fig-ex1}
and~\ref{fig-ex2}.
The first nine charts were created using the first notion of randomness, with varying values of $p$ and $s$.  The last three charts were created using the second notion of randomness, with varying values of $p$. The unit length used in all of the tests was 30.
Each chart has been set to use the same vertical axis for easy visual comparison between charts. 

Our experiments provide evidence that
our algorithm produces a drawing that is 
anywhere from slightly smaller up to eight times smaller than the previous two algorithms. 
One conclusion that can be drawn from these results is
that the size of our drawing is not affected greatly by the parameters, while the sizes of the other two drawings increase quickly as the number of parallel splits increases.

\begin{figure}[htbp]
\begin{center}
\begin{tabular}{cc}
\includegraphics[scale=0.42]{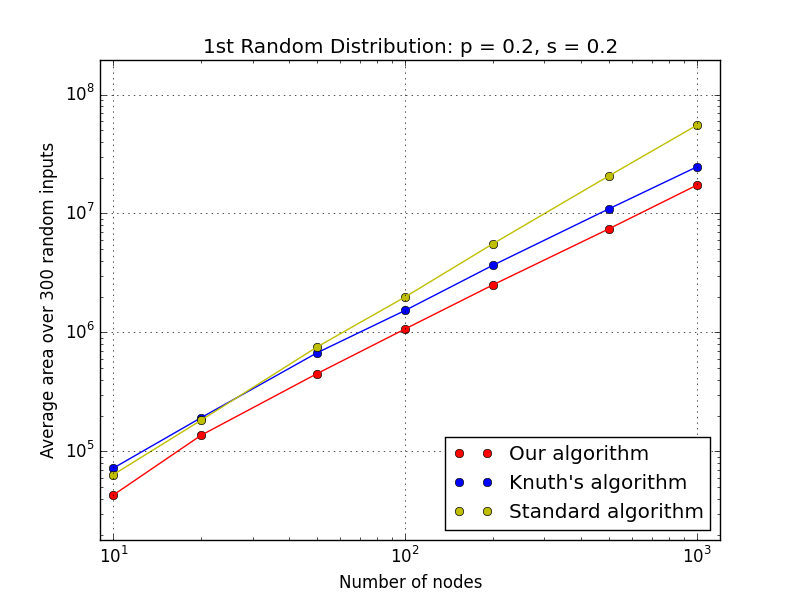} & \includegraphics[scale=0.42]{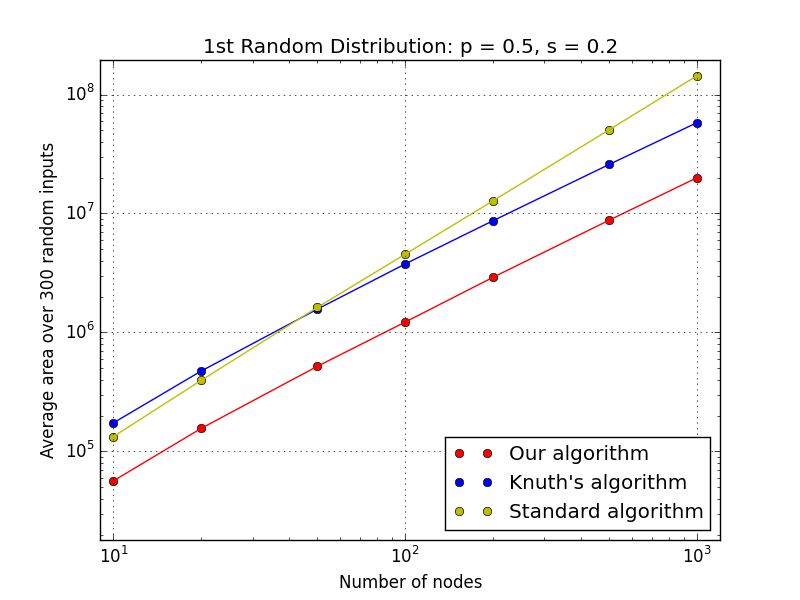} \\
\includegraphics[scale=0.42]{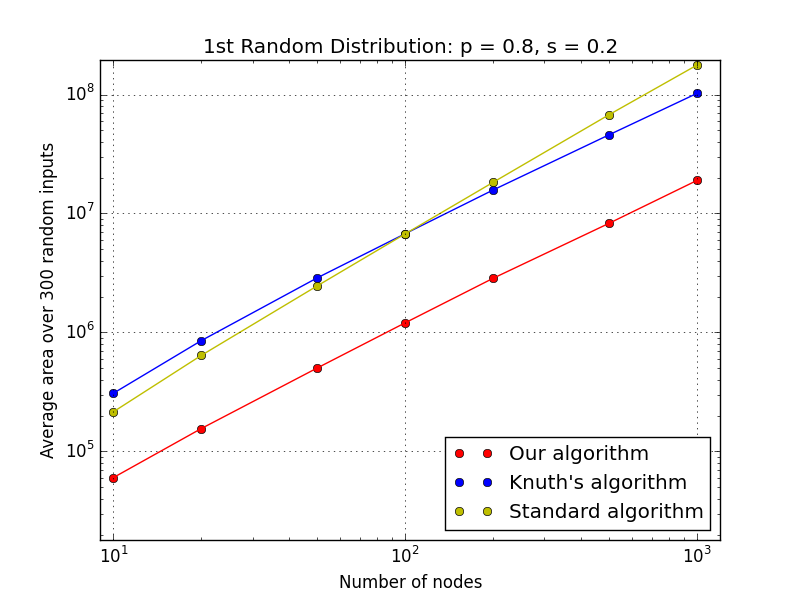} & \includegraphics[scale=0.42]{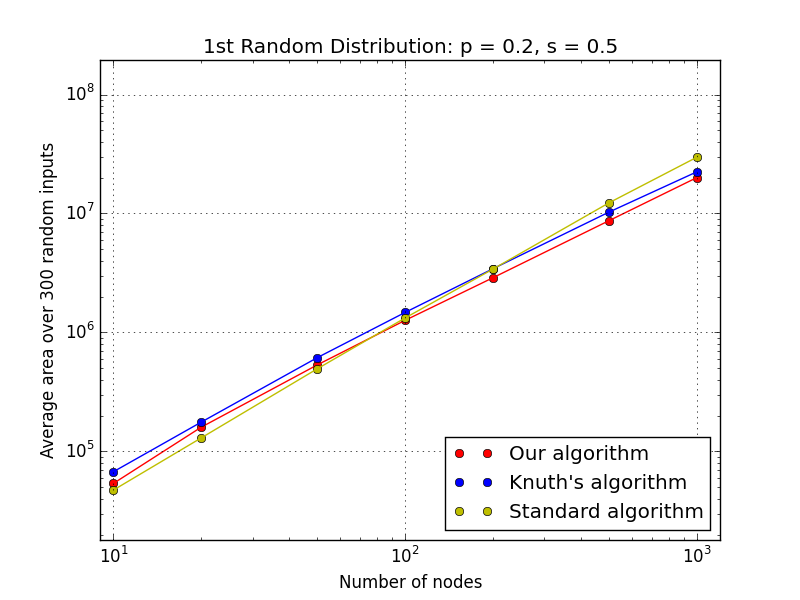} \\
\includegraphics[scale=0.42]{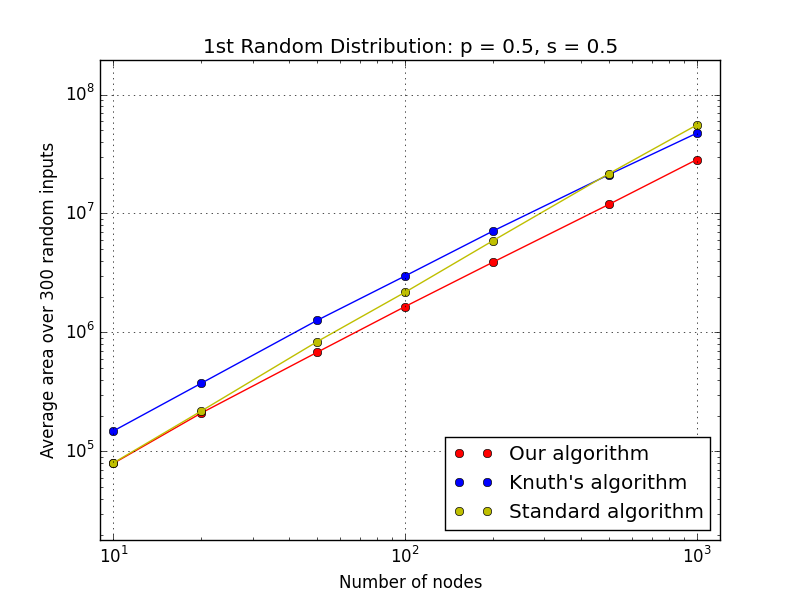} & \includegraphics[scale=0.42]{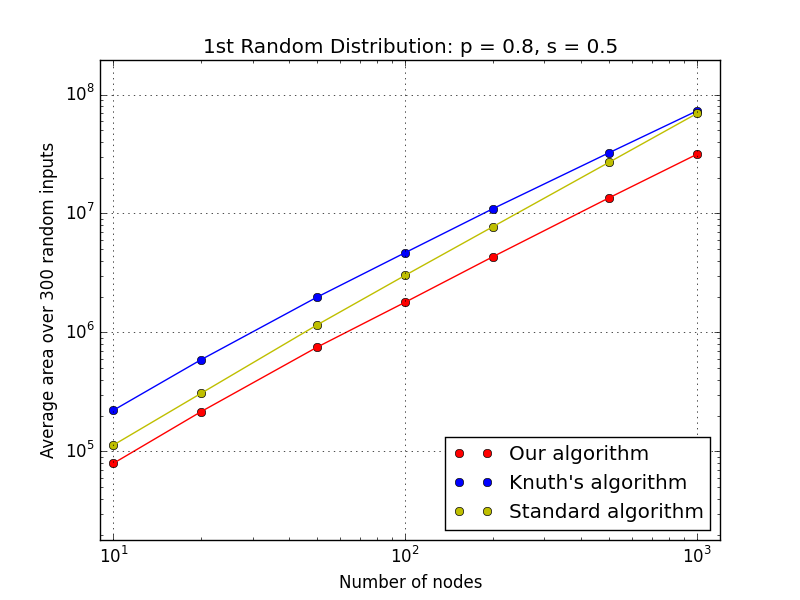} \\
\end{tabular}
\end{center}
\caption{Some results of our experiments.}
\label{fig-ex1}
\end{figure}

\begin{figure}[htbp]
\begin{center}
\begin{tabular}{cc}
\includegraphics[scale=0.42]{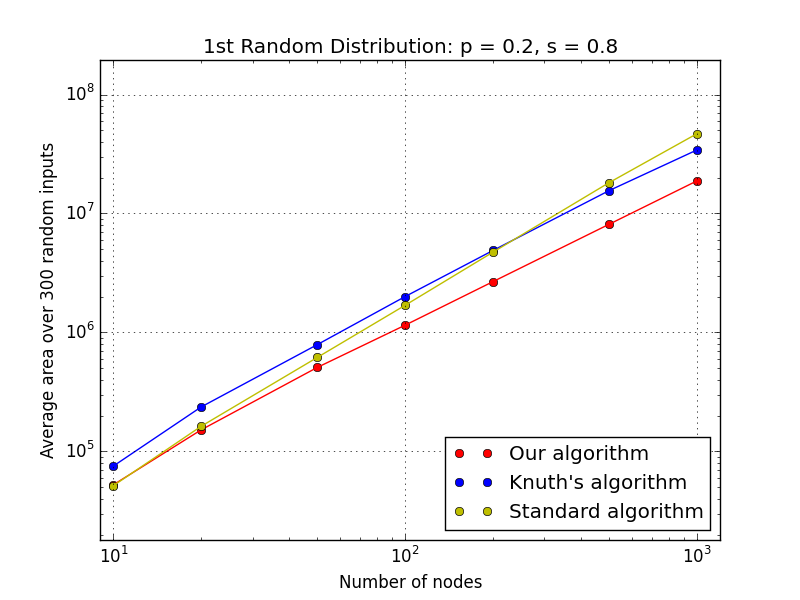} & \includegraphics[scale=0.42]{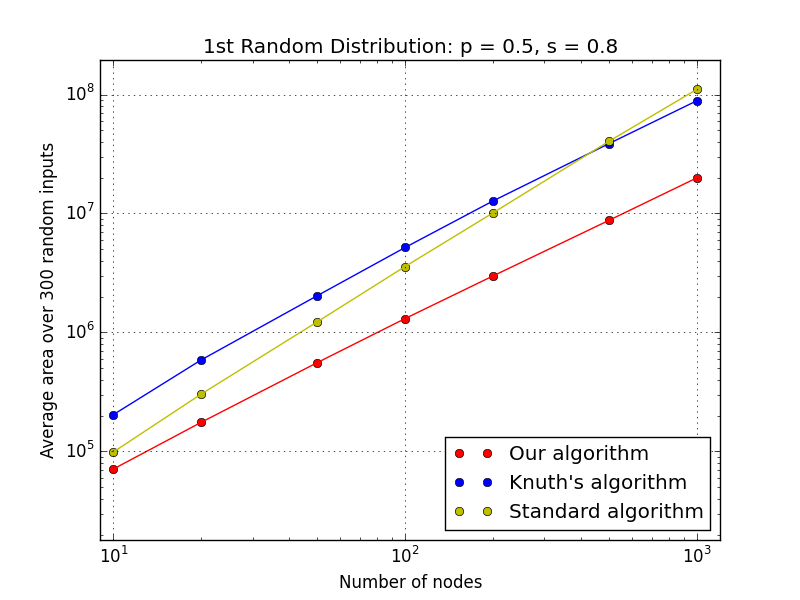} \\
\includegraphics[scale=0.42]{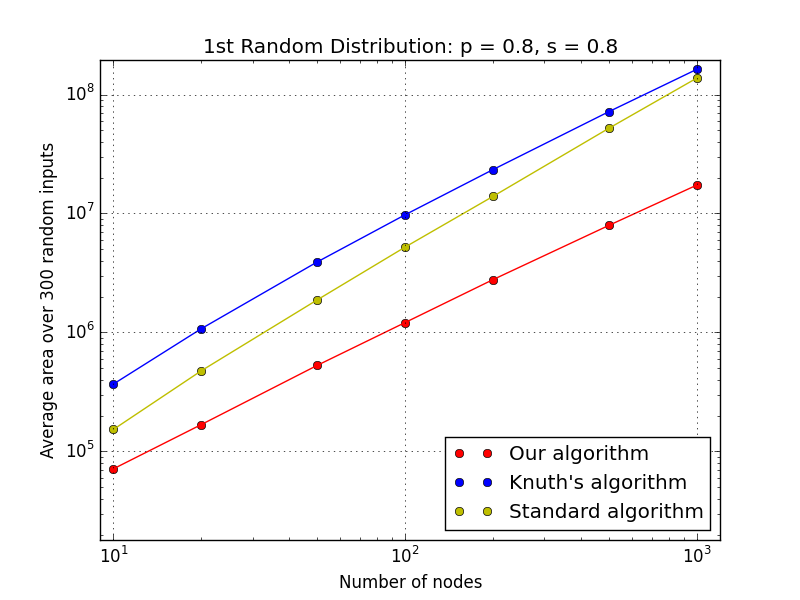} & \includegraphics[scale=0.42]{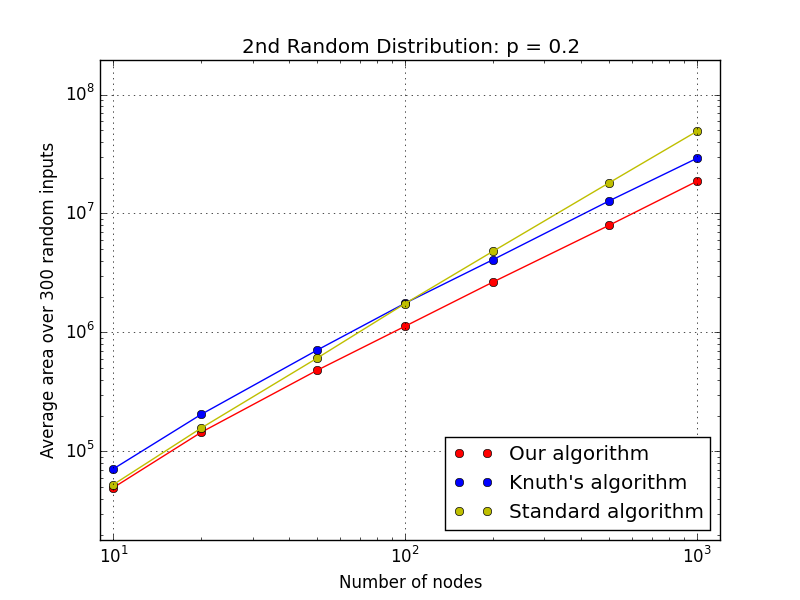} \\
\includegraphics[scale=0.42]{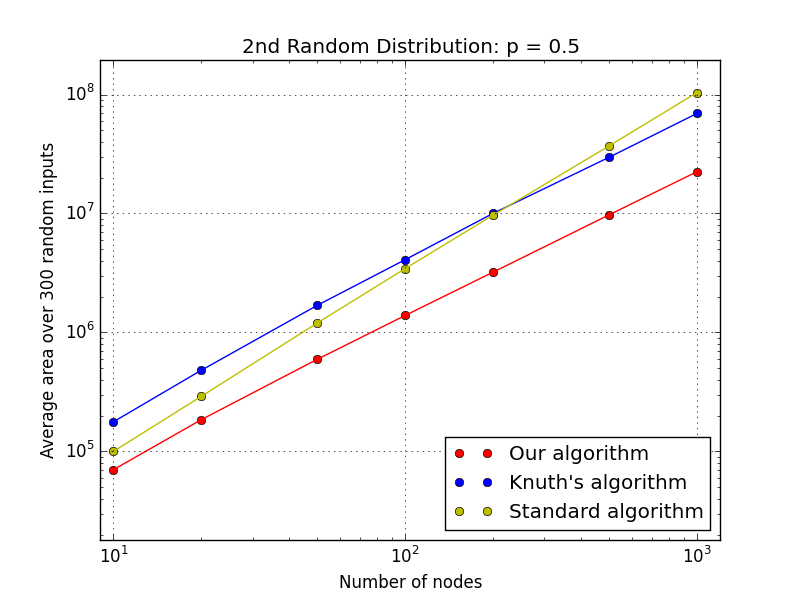} & \includegraphics[scale=0.42]{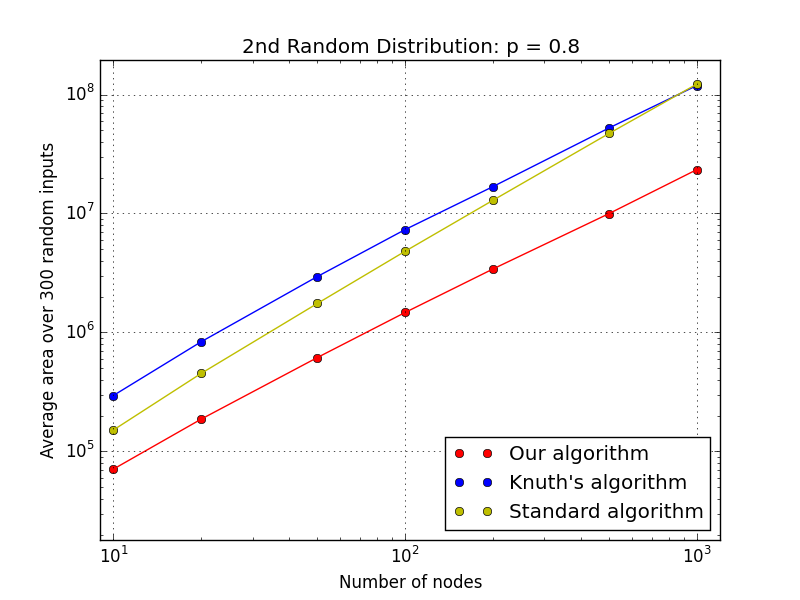}
\end{tabular}
\end{center}
\caption{Some additional results of our experiments.}
\label{fig-ex2}
\end{figure}

\subsection*{Acknowledgments} 
This research was supported in part by
the National Science Foundation under grants 1011840
and 1228639.
This article also reports on work supported by the Defense Advanced
Research Projects Agency (DARPA) under agreement no.~AFRL FA8750-15-2-0092.
The views expressed are those of the authors and do not reflect the
official policy or position of the Department of Defense
or the U.S.~Government.

{\raggedright 
\bibliographystyle{abbrv} 
\bibliography{flowcharts} 
}

\clearpage
\begin{appendix}
\section{Appendix}
In this appendix,
we discuss in depth the implementation details for 
constructing the single-row drawing in linear time. 
The implementation occurs in three stages: initialization, preprocessing, and drawing.  
The input to the algorithm is a string representing a preorder traversal of the tree representing a degree-three series-parallel graph.

In the first stage, a decomposition tree is created in a linear-time pass over the input string. 
The second stage consists of determining for each node $v$ the type of graph rooted at $v$ based off of the two types of graphs in Figure~\ref{fig5}. 
Recall that we named the first type of graph a broad graph, and the second type of graph a pinched graph.

The subtree at $v$ describes a broad graph if it indicates a parallel composition or if it indicates a series composition and at most one of its children is broad.  
Otherwise, the graph is pinched.  
In either case, we store pointers in $v$ that point to the nodes in the subtree of $v$ that are the roots of the subtrees for each component. 
In a broad graph, these pointers are for the subgraphs of sizes
$s_1$, $s_2$, $p_1$, and $p_2$, respectively.  
In a pinched graph, we add pointers for 
the subgraphs of sizes $s_1$, $s_2$, $p_1$, $p_2$, $p_3$, $p_4$, and $p_5$,
respectively. We begin at the leaves and propagate the types of graphs up the tree so that only a constant amount of work is being done at each node to check the children and set pointers, so the second stage is also a 
linear-time operation.

Having these pointers for each node is advantageous in the third stage so that we don't have to perform a search to find the different sections. Then as we construct each section, we store an internal representation that can be shifted to any origin point by changing a single offset value, so each shift takes constant time. We perform the actual drawing only once we know where every component belongs.

As an example, let us
describe the case in which $p_3$ is the largest. For pinched graphs, 
we first construct drawings of the subgraphs of sizes $s_1$ and $s_2$. 
We then draw the subgraphs of sizes $p_1$ and $p_2$ in parallel at the end of section $s_1$ using the left-to-right drawing style given in 
Figure~\ref{fig4}. 

Next we construct the drawings of the subgraphs of size $p_4$ and $p_5$ 
using the right-to-left, upside-down drawing style given in Figure~\ref{fig5}.
We need to construct these before constructing the drawing of the
subgraph of size $p_3$, because we may need to increase the height of this
drawing, since it needs to be taller than $p_4$ and $p_5$.  
Once we draw the subgraph of size $p_3$ with the correct height, 
we then shift the drawings for the subgraphs of size $p_4$ and $p_5$ into their
proper positions. To finish the drawing, 
we shift the drawing of 
the subgraph of size $s_2$ to the end of the $p_4$ and $p_5$ section.

Finally, we draw every node at its proper location. 
Since there is no backtracking, each node is preprocessed once, constructed once, shifted once to fit with the other drawings in the same step, and drawn once.
Thus, the overall time complexity is $O(n)$.

\end{appendix}

\end{document}